\documentclass[sigplan,screen,nonacm]{acmart}

\acmConference[LICS'22]{LICS}{2022}{online?}

\usepackage[utf8]{inputenc}
\usepackage[UKenglish]{babel}
\usepackage{mathtools}
\usepackage{amsthm}
\usepackage{bm} 
\usepackage{enumitem}
\usepackage{mathpartir}

\usepackage{tikz}
\usetikzlibrary{cd}
\usetikzlibrary{calc}
\usetikzlibrary{decorations.pathmorphing}

\tikzset{commutative diagrams/kl/.style={rightsquigarrow}}
\newcounter{sarrow}
\newcommand\xrsquigarrow[1]{%
\stepcounter{sarrow}%
\mathrel{\begin{tikzpicture}[baseline= {( $ (current bounding box.south) + (0,-0.3ex) $ )}]
\node[inner sep=1ex] (\thesarrow) {$\scriptstyle #1$};
\path[draw,<-,decorate,
  decoration={zigzag,amplitude=0.7pt,segment length=1.2mm,pre=lineto,pre length=4pt}] 
    (\thesarrow.south east) -- (\thesarrow.south west);
\end{tikzpicture}}%
}

\usepackage{cleveref}
\usepackage{todonotes}
\definecolor{mycolour1}{HTML}{C9D5AE}
\definecolor{mycolour2}{HTML}{AECAD5}
\definecolor{mycolour3}{HTML}{D5AECA}

\newcommand{\opens}{\mathcal O}
\newcommand{\sets}{\ensuremath{\mathsf{Set}}}
\newcommand{\inj}{\ensuremath{\mathsf{Inj}}}
\newcommand{\cat}[1]{\ensuremath{\mathsf{#1}}}

\newcommand{\Tcom}{T^\dagger}   
\newcommand{\klcomp}{\mathbin{\circledcirc}}

\newcommand{\id}{\mathsf{id}}
\newcommand{\Kl}{\mathsf{Kl}}
\newcommand{\cop}{\mathsf{copy}}
\newcommand{\del}{\mathsf{del}}
\newcommand{\force}{\mathsf{force}}
\newcommand{\thunk}{\mathsf{thunk}}
\newcommand{\samp}{\mathsf{samp}}

\DeclareMathOperator{\ob}{ob}
\newcommand{\CComon}{\mathsf{CComon}}
\DeclareMathOperator{\colim}{colim}
\DeclareMathOperator{\im}{im}

\newcommand{\bB}{\ensuremath{\mathbb B}}
\newcommand{\bC}{\ensuremath{\mathbb C}}

\newcommand{\bK}{\ensuremath{\mathbb K}}

\newcommand{\bN}{\ensuremath{\mathbb N}}

\newcommand{\bR}{\ensuremath{\mathbb R}}
\newcommand{\bS}{\ensuremath{\mathbb S}}
\newcommand{\bT}{\ensuremath{\mathbb T}}

\newcommand{\MCC}{\ensuremath{\mathcal C}}

\newcommand{\MCK}{\ensuremath{\mathcal K}}

\newcommand{\MCM}{\ensuremath{\mathcal M}}

\newcommand{\MCP}{\ensuremath{\mathcal P}}

\AtEndPreamble{
  \theoremstyle{acmdefinition}
  \newtheorem{remark}[theorem]{Remark}
}

\title{Probability monads with submonads of deterministic states --  Extended version}

\author{Sean Moss}
\affiliation{%
  \institution{University of Oxford}
  \department{Department of Computer Science}
  \city{Oxford}
  \country{UK}}
\email{sean.moss@cs.ox.ac.uk}
\author{Paolo Perrone}
\affiliation{%
  \institution{University of Oxford}
  \department{Department of Computer Science}
  \city{Oxford}
  \country{UK}}
\email{paolo.perrone@cs.ox.ac.uk}

\begin{document}

\begin{abstract}
Probability theory can be studied synthetically as the computational effect embodied by a commutative monad. In the recently proposed Markov categories, one works with an abstraction of the Kleisli category and then defines deterministic morphisms equationally in terms of copying and discarding. The resulting difference between `pure' and `deterministic' leads us to investigate the `sober' objects for a probability monad, for which the two concepts coincide. We propose natural conditions on a probability monad which allow us to identify the sober objects and define an idempotent sobrification functor. Our framework applies to many examples of interest, including the Giry monad on measurable spaces, and allows us to sharpen a previously given version of de Finetti's theorem for Markov categories.
\end{abstract}

\maketitle

This is an extended version of the paper accepted for the Logic In Computer Science (LICS) conference 2022. 
In this document we include more mathematical details, including all the proofs, of the statements and constructions given in the published version.

\paragraph{About citing this work.} 
All the definitions, propositions, and theorems appearing in the published version also appear here, with the same numbering as in the published version. There is one result here, \Cref{gfg}, not present in the published version.
The numbering of particular equations is however inevitably different between the two versions. 
Because of this, if future readers need to refer to any of the equations contained here, we recommend them to refer to the corresponding definition or theorem instead.

\section{Introduction}

This paper is about different models of `abstract' or `synthetic' probability theories.
Such theories talk about both \emph{deterministic quantities} and \emph{random quantities} (i.e.\ \emph{random variables}).
The difference is analogous to the distinction between \emph{values} and \emph{computations} in the semantics of programming languages.
Indeed, particular non-standard models of synthetic probability have proved useful in applications to probabilistic programming:
\begin{itemize}
\item Quasi-Borel spaces (\cite{heunen-kammar-staton-yang-a-convenient-category-for-higher-order-probability-theory}, \cite{scibior-et-al-denotational-validation-of-higher-order-bayesian-inference}) model Kock's synthetic measure theory \cite{kock-commutative-monads-as-a-theory-of-distributions}.
  Unlike traditional foundations for probability in measurable spaces, they are well-suited to higher-order data.
\item While naive handling of conditional probabilities can lead to paradoxes \cite{jacobs-paradoxes-of-probabilistic-programming}, it was shown in \cite{stein-staton-compositional-semantics-for-probabilistic-programs-with-exact-conditioning} that in more restrictive models of probability `exact conditioning' can be given a consistent meaning.
  Fritz's Markov categories \cite{fritz2019synthetic} were used to formulate the result.
\end{itemize}

A useful way to present a model of probability is with a \emph{commutative monad} \cite{kock-commutative-monads-as-a-theory-of-distributions}.
This gives the link to the semantics of more general programming effects, since commutative monads are a special case of the \emph{strong monads} used by Moggi for the semantics of call-by-value languages \cite{moggi-notions-of-computation-and-monads}.
A potentially strange aspect of monadic semantics is that it effectively gives both values and computations separately, with no need for one to be a subset of the other.
In terms of probability, we are led to ask whether the deterministic quantities can be characterized as certain well-behaved random quantities.
Two useful criteria, relevant to both probability and more general computation are the following, described informally.
\begin{enumerate}
\item \emph{Discardable}: A computation that has negligible effect if its result is not used, so may be safely discarded.
\item \emph{Copyable}: A computation which can be run once and have its output used twice instead of being run twice.
\end{enumerate}
Unlike monadic semantics, in a Markov category only the random quantities are explicitly given.
All the quantities in a Markov category are discardable, and then deterministic quantities are \emph{defined} to be those that are also copyable.

Having deterministic quantities characterized equationally has enabled many interesting results from traditional probability theory to be expressed and proved synthetically.
However, it is often the case that we wish to preserve some connection with the `pure' quantities, i.e.\ those morphisms in the original category hosting a monad for probability.
This is because that category is where we have explicit descriptions of the objects.
The goal of this paper is to elucidate such a connection for several examples of interest.
More specifically, we address whether probability monads of interest admit `sobrification' submonads, i.e.\ a universal way of replacing each object with one for which pure quantity coincides with deterministic quantity.

In the remainder of this introduction, we provide some informal explanation of the background material on probability monads and a high-level picture of the development in this paper.
Our main contribution rests in packaging up properties of certain probability monads into the \emph{observationality} and \emph{S-observationality} conditions, showing that these are indeed satisfied by examples and interest, and showing that these conditions abstractly imply the desired results on sobrification.

\subsection{Probability monads}

The canonical mathematical model of probability theory is in the category $\cat{Meas}$ of measurable spaces.
A measurable space is a pair $(X,\Sigma_X)$ where $X$ is a set and $\Sigma_X \subseteq \MCP(X)$ is a collection of subsets of $X$ containing $\emptyset$ and $X$ and closed under countable unions and complements.
Morphisms $(X,\Sigma_X) \to (Y,\Sigma_Y)$, or `measurable functions', are functions $f : X \to Y$ such that $f^{-1}(E) \in \Sigma_X$ for every $E \in \Sigma_Y$.
It is necessary to consider measurable spaces rather than just sets because there are foundational problems with constructing probability distributions that assign probability to \emph{every} subset of a given set $X$.
Instead, we equip $X$ with a family $\Sigma_X$ of `measurable sets', and require a \emph{probability measure} to be a function $m : \Sigma_X \to [0,1]$ satisfying $m(X) = 1$, $m(\emptyset) = 0$ and countable additivity.

Interestingly, the set $PX = P(X,\Sigma_X)$ of probability measures on a measurable space $(X,\Sigma_X)$ can itself be considered as a measurable space.
In fact, $P$ is the functor part of the \emph{Giry monad} on $\cat{Meas}$ (see \Cref{girymonad}). This allows us to consider `probability measures on the set of probability measures'.
The Giry monad is a strong monad, and indeed a commutative monad, so following Moggi \cite{moggi-notions-of-computation-and-monads} we can use it to interpret a first-order call-by-value language.

The basic setting in this paper is that of a category $\bC$ with finite products and a monad $T$ on it.
The idea is that a monad includes an assignment $A \mapsto TA$, where $TA$ is the object of `distributions' or `measures' on $A$.
Hence, for any pair of objects $A,X \in \ob\bC$, we can consider morphisms $A \to X$ to be `pure' functions and morphisms $A \to TX$ to be `stochastic functions'.
In the language of computation, the latter would be a computation that produces an $X$.

\subsection{Thunkable morphisms}

Besides the aforementioned properties of copyability and discardability, there is an additional property that makes a monadic computation look `ordinary': \emph{thunkability}.
A \emph{thunk} is a computation that has been `frozen'.
In probability, the `thunk' of a state $x$ is a `Dirac delta' probability distribution $\delta_x$ which assigns probability $1$ to $x$ and $0$ to everything else.
In other words, sampling from $\delta_x$ almost surely returns $x$.
Informally, a program $M : A$ is \emph{thunkable} if it satisfies the equation
\begin{displaymath}
  \lambda \_. M = \mathrm{let}\, x \leftarrow M \, \mathrm{in}\, \lambda \_. x : 1 \to A.
\end{displaymath}
In terms of probability, we think of $M$ as a distribution, the difference between the two sides is that the left-hand side returns a thunk that samples anew from $M$ every time it is run, but the right hand side creates a thunk by sampling from $M$ once and for all then wrapping the result.

From the point of view of probability theory, thunkable morphisms are those that commute with forming Dirac deltas. This can be seen as related to determinism, since any stochastic map which `spreads' the mass of a measure, from a single point to several ones, cannot possibly commute with forming deltas.
(See also \cite[Remark 3.11]{fritz2020representable} for additional context.)

It therefore seems that thunkability is yet another property that sets pure computations apart from the other ones, and so it is interesting to study the relationship between purity, thunkability, copyability, and discardability \cite{fuhrmann-direct-models-of-the-computational-lambda-calculus,fuhrmann-varieties-of-effects}. We recap these properties and their relationships in \Cref{structured}.
We show that in nice situations thunkable morphisms are encoded by a submonad of the original monad, often idempotent, which one can think of as measuring the extent to which a generic thunkable morphism fails to be pure.

\subsection{Observations and de Finetti's theorem}

In the practice of probabilistic programming, as well as in statistics, one is sometimes given two random variables and has to test whether they follow the same distribution or not. The problem is not an easy one. 
In the best case one can draw independent samples from them, and compare the (random) sequences obtained by the repeated draws. 
This idea is also reflected by the famous \emph{de Finetti theorem}~\cite{de-finetti-foresight-its-logical-laws-its-subjective-sources}, which roughly says that random probability measures correspond bijectively to exchangeable random sequences.
In this work we make mathematically precise, in terms of monads, the intuition that \emph{random variables are tested by repeated draws}, or more generally, that in some contexts, effectful computations can be compared by taking repeated independent runs.
In terms of a program $M : 1 \to (1 \to A)$, we test with the program contexts
\begin{align*}
  \MCC_n[M] \coloneqq{} & \mathrm{let}\, z \leftarrow M()\, \mathrm{in}\, \mathrm{let}\, x_1 \leftarrow z()\, \mathrm{in}\,\ldots \\
  & \mathrm{let}\, x_n \leftarrow z()\, \mathrm{in}\, \mathrm{return}\, (x_1,\ldots,x_n)
\end{align*}
where $\MCC_n[M] : A^n$, for $n \in \bN$.
To see why repeated sampling is necessary, consider the following programs
\begin{mathpar}
  M_1 \coloneqq \lambda \_.\lambda \_.\mathrm{or}(\mathrm{true},\mathrm{false})
  \and
  M_2 \coloneqq \lambda \_.\mathrm{or}(\lambda \_.\mathrm{true} , \lambda \_.\mathrm{false})
\end{mathpar}
of type $1 \to (1 \to \mathrm{bool})$ where $\mathrm{or}$ is a non-deterministic choice.
Then $\MCC_1[M_1]$ and $\MCC_1[M_2]$ can each evaluate to both $\mathrm{true}$ and $\mathrm{false}$, but $\MCC_2[M_1]$ can evaluate to only two of the four possibilities and $\MCC_2[M_2]$ can evaluate to any of them.

We call monads that exhibit this property \emph{observational}.
While this property is often considered typical of probabilistic contexts, somewhat surprisingly it can also hold for monads which are not strictly about probability. For example, the lower Vietoris monad (or Hoare powerdomain) is observational (\Cref{Hobserv}), and it generally does not encode randomness, but rather, nondeterminism. 
For observational monads, as we prove in \Cref{main}, every thunkable morphism is deterministic.

\subsection{Outline}

In \Cref{monads} we recall the basic concepts of monad theory, with a view on the probability case. We in particular look at the \emph{equalizing requirement}, which is one of the less known concepts which is of great relevance for this work, and we sketch how a monad interacts with the products of a monoidal category, in order to form `joint states'. 

In \Cref{structured} we look in detail at the structures that select the different versions of effectful (in particular, random) and noneffectful computation. 
We recall (\Cref{inclusions}) that every pure morphism is thunkable and every thunkable morphism is copyable and discardable, but not the other way around (as counterexamples show).

In \Cref{submonad} we construct for every monad $T$ a submonad $D$ whose Kleisli morphisms are the thunkable morphisms of $T$. 
In \Cref{sober} we define sober objects as those for which every thunkable morphism is pure, and look at the relationship with the submonad $D$. We show that in several cases, such as for the Giry monad, the submonad $D$ is idempotent, and that sober objects generalize sober topological spaces.

In \Cref{observational} we define observational monads, those for which computations can be tested by repeated independent runs. We look at the connection with ground types, the ones that are directly observable (e.g.~by the user or experimenter), and we give technical conditions to show observationality of monads.
In \Cref{detmaps} we prove what could be considered the main result of this work (\Cref{main}), namely that for observational monads, deterministic morphisms are the same as thunkable morphisms. Therefore, in that case the submonad $D$ equivalently encodes deterministic morphisms.

In \Cref{definetti} we connect the notion of observationality with de Finetti's theorem. In particular, we show that for observational monad one can sharpen the known version of the synthetic de Finetti theorem for Markov categories \cite{definetti-markov}.

Finally, in \Cref{examples} we prove observationality for three apparently very different monads: the Giry monad, the lower Vietoris monad of nondeterminism, and the  
monad of name generation. 
We give concluding remarks in \Cref{conclusion}.

\section{Background on monads}\label{monads}

Recall \cite{maclane-cwm1998} that a \emph{monad} on a category $\bC$ is a triple $\bT = (T,\eta,\mu)$ (subsequently also just denoted by $T$) where
\begin{enumerate}
\item $T$ is a functor $\bC \to \bC$,
\item $\eta$ is a natural transformation $\id \Rightarrow T$,
\item $\mu$ is a natural transformation $TT \Rightarrow T$,
\end{enumerate}
satisfying $\mu_A \circ \eta_{TA} = 1_{TA} = \mu_A \circ T\eta_A$ and $\mu_A\circ\mu_{TA}=\mu_A\circ T\mu_A$ for all $A \in \bC$.

Monads are used in denotational semantics to model a distinction between \emph{pure values} and \emph{effectful computations} \cite{moggi-notions-of-computation-and-monads}.
If morphisms $A \to X$ are ordinary `$X$-valued functions', then morphisms $A \to TX$ are `$X$-producing computations'.
For the purposes of this work, we are mainly interested in those monads involving probability, i.e.~for which a morphism $A\to TX$ can be seen as a \emph{stochastic} map from $A$ to $X$, involving random chance.

\begin{example}\label{distmonad}
  Consider the category $\sets$ whose objects are sets and morphisms are functions.
  The \emph{distribution monad} on $\sets$ is the monad $T$ with $TX = \{ \pi \in [0,1]^X : \sum_{x \in X} \pi(x) = 1 \}$, $\eta_X(x) = \lambda x'.\llbracket x = x_0 \rrbracket$ (i.e.~it forms a ``delta at $x$'') and
  \begin{displaymath}
    \mu_X(\rho)(x) = \sum_{\pi \in TX} \rho(\pi)\times\pi(x).
  \end{displaymath}
\end{example}

\begin{example}\label{girymonad}
 Consider the category $\cat{Meas}$ whose objects are measurable sets and morphisms are measurable functions. 
 The \emph{Giry monad} $P$ \cite{giry} consists of
 \begin{itemize}
  \item The functor $P$ assigning to each measurable space $X$ the set $PX$ of probability measures over $X$, equipped with the coarsest $\sigma$-algebra which makes the evaluation of measures measurable;
  \item The natural transformation of components $\eta:X\to PX$ assigning to each point $x\in X$ the ``Dirac delta'' measure $\delta_x$, such that for every measurable $A\subseteq X$, $\delta_x(A) = 1$ if $x \in A$ and $\delta_x(A) = 0$ otherwise.
  \item The natural transformation of components $\mu:PPX\to PX$ which, analogously to the case of the distribution monad, assigns to each measure $\rho\in PPX$ the ``mixture'' measure $\mu(\rho)$, such that for every measurable $A\subseteq X$, 
  $$
  \mu(\rho)(A) = \int_{PX} p(A) \, \mu(d\rho) .
  $$
  
 \end{itemize}
 
 On $\cat{Meas}$ we can analogously define the monad $M$ of \emph{subprobability} measures, where now instead of having $p(X)=1$ we only require $0\le p(X) \le 1$.
 
 For more information, we refer to \cite{giry}. For an introduction to the concepts, see \cite[Chapter~1]{thesispaolo}, \cite[Section~6]{fritz2018monads} and \cite{panangaden-labelled-markov-processes}.
\end{example}

In this view of a monad, $\eta_A : A \to TA$ is the computation which, given $a \in A$, just returns $a$.
The rest of the monad data tells us how to sequence computations together.
The \emph{Kleisli category} $\Kl(T)$ is the category with the same objects as $\bC$ but homsets the $\Kl(T)(A,B)$ is in bijection with $\bC(A,TB)$.
We explicitly denote this bijection with
\begin{mathpar}
  (-)^\sharp : \Kl(T)(A,B) \to \bC(A,TB)
  \and
  (-)^\flat : \bC(A,TB) \to \Kl(T)(A,B).
\end{mathpar}
The composition in $\Kl(T)(A,B)$, denoted by $\klcomp$, is given by
\begin{displaymath}
  (g \klcomp f)^\sharp = \mu \circ T(g^\sharp) \circ f^\sharp
\end{displaymath}
with the identity maps $(1_A)^\sharp = \eta_A$.
Since the morphisms of $\Kl(T)$ can easily be confused with those of $\bC$, in diagrams we use an ordinary arrow $f : A \to TB$ when considering morphisms in $\bC$ and a wavy arrow $f^\flat : A \rightsquigarrow B$ for morphisms in $\Kl(T)$.
The bijection above gives an adjunction between $\bC$ and $\Kl(T)$ called the \emph{Kleisli adjunction}.

\begin{example}
 For the distribution monad on \sets, a Kleisli morphism is a function $X \to TY$, i.e.~an $X$-indexed family of probability distributions on $Y$, or a function on $X$ whose output is random, also called a \emph{finite probability kernel}. We can denote its entries by $k(y|x)$, interpreted as the conditional probability of obtaining output $y$ from the input $x$. 
 
 The Kleisli composition of $k:X \to TY$ and $h:Y\to TZ$ returns the kernel $h\klcomp k:X\to TZ$ given by
 $$
 h\klcomp k (z|x) = \sum_{y\in Y} h(z|y) \, h(y|x) .
 $$
 also known as the \emph{Chapman-Kolmogorov} composition of probability kernels.
\end{example}

\begin{example}
 We denote the Kleisli categories of $P$ and $M$ by $\cat{Stoch}$ and $\cat{SubStoch}$, respectively.
 Given measurable spaces $X$ and $Y$, a morphism $k:X\rightsquigarrow Y$ of $\cat{Stoch}$ is a \emph{Markov kernel} (or \emph{stochastic map}) $k$ from $X$ to $Y$, i.e.~either a measurable function $X\to PY$,or equivalently a map 
 $$
 \begin{tikzcd}[row sep=0]
  X\times\Sigma_Y\ar{r}{k} & {[0,1]} \\
  (x,B) \ar[mapsto]{r} & k(B|x)
 \end{tikzcd}
 $$ such that
 \begin{itemize}
  \item for each measurable subset $B\subseteq Y$, the assignment $x\mapsto k(B|x)$ is measurable;
  \item for each $x\in X$, the assignment $B\mapsto k(B|x)$ is a subprobability measure (i.e.~$k(Y|x) = 1$). 
 \end{itemize}
 The Kleisli composition is the continuous analogue of the Chapman-Kolmogorov formula, it is given by the integral
 $$
 h\klcomp k(C|x) = \int_{Y} h(C|y)\,k(dy|x) .
 $$
 
 The morphisms of $\cat{SubStoch}$, called \emph{substochastic maps}, are defined similarly, except that $0 \le k(Y|x) \le 1$.
\end{example}

In some sense, the unit of the monad allows to transport morphisms of the base category into the Kleisli category.:
\begin{definition}\label{defpure}
 A morphism $f:X\rightsquigarrow Y$ of $\Kl(T)$ is called \emph{pure} if its counterpart $f^\sharp:X\to TY$ is in the form 
$$
\begin{tikzcd}
X \ar{r}{g} & Y \ar{r}{\eta} & TY
\end{tikzcd}
$$
for some $g:X\to Y$ of $\bC$. We call $f$ \emph{uniquely pure} if $f^\sharp$ can be written as $\eta\circ g$ for a unique $g$.
\end{definition}

Pure morphisms are, in some sense, those that ``come from the base category'', or ``do not really use the monad''. 

\begin{example}
For the Giry monad of probability measures, a pure 
morphism $f:X\rightsquigarrow Y$ is a kernel assigning to each point of $x$ a Dirac measure on $Y$, in a measurable way. So, in some sense, it is simply a measurable function $X\to Y$. Note that, if the $\sigma$-algebra of $Y$ does not separate points, different measurable functions $X\to Y$ might define the same kernel, and so in general a pure morphism $f:X\rightsquigarrow Y$ is not uniquely pure. An example of that will given in \Cref{notinjective}.
\end{example}

\subsection{Unit fork and equalizing requirement}

Recall that a \emph{fork} is a diagram 
$$
\begin{tikzcd}
 A \ar{r}{f} & B \ar[shift left]{r}{g} \ar[shift right, swap]{r}{h} & C
\end{tikzcd}
$$
such that $g\circ f = h\circ f$ (but not necessarily $g=h$). One can view an equalizer as a universal fork.

Given a monad $T$, for each object $X$, the unit of the monad forms the following fork,
\begin{equation}\label{unitfork}
 \begin{tikzcd}
  X \ar{r}{\eta} & TX \ar[shift left]{r}{\eta} \ar[shift right]{r}[swap]{T\eta} & TTX
 \end{tikzcd}
\end{equation}
which is indeed a fork by naturality of $\eta$. We call the diagram \eqref{unitfork} the \emph{unit fork} at $X$. 

The monads for which the unit fork is an equalizer for every $X$ are said to satisfy the \emph{equalizing requirement}
\cite[Section~4]{moggi-notions-of-computation-and-monads}.

\begin{example}
 The distribution monad on $\sets$ satisfies the equalizing requirement. 
 Indeed, given $p\in TX$, the two distributions $\eta(p)$ and $T\eta(p)\in TTX$ are respectively, a delta peaked at $p$, and a convex combination of deltas at points $x$ with coefficients $p(x)$. These are equal if and only if $p$ is itself a delta at some point $x$. 
\end{example}

\begin{example}\label{notinjective}
 (This example comes from \cite[Example~10.5]{fritz2019synthetic}.)
 The Giry monad on $\cat{Meas}$ does not satisfy the equalizing requirement. 
 Let $X$ be the 2-point space $\{x,x'\}$, equipped with the codiscrete sigma-algebra (i.e.~the only measurable sets are the empty set and $X$ itself). 
 Then the measures $\delta_x$ and $\delta_{x'}$ are equal, even if the points $x$ and $x'$ are not.
 As the unit $\eta$ (i.e.~$\delta$) is not injective, it cannot be an equalizer. 
\end{example}

\begin{remark}\label{TXeq}
The unit fork for objects in the form $TX$ is always an equalizer, in fact a split one.
$$
\begin{tikzcd}
 TX \ar{r}{\eta} & TTX \ar[bend left]{l}{\mu} \ar[shift left]{r}{\eta} \ar[shift right]{r}[swap]{T\eta} & TTTX \ar[shift left, bend left]{l}{T\mu} 
\end{tikzcd}
$$
\end{remark}

\subsection{Monads on monoidal categories}

A \emph{symmetric monoidal category} (SMC) is a category equipped with a ``tensor product'', i.e.~a binary functor $\otimes:\bC\times\bC\to\bC$, and a ``unit'' object $I$, together with isomorphisms 
\begin{mathpar}
  X\otimes I \cong X \cong I\otimes X ,
  \and
  (X\otimes Y) \otimes Z \cong X \otimes (Y\otimes Z) ,
  \and
  \tau_{X,Y} : X\otimes Y \cong Y\otimes X
\end{mathpar}
satisfying appropriate coherence conditions and $\tau_{X,Y} \circ \tau_{Y,X} = 1$.
By a coherence theorem we can calculate as though the first three are actually \emph{identities} --- see \cite{maclane-cwm1998} for precise definitions and details.
The interpretation is that given objects $A,B,C,D$ and morphisms $f:A\to B$ and $g:X\to Y$, we can form new objects $A\otimes X$ and $B\otimes Y$ of ``joint states'', and the morphism $f\otimes g: A\otimes X\to B\otimes Y$, which as a process it consists of ``executing $f$ and $g$ independently, in parallel''.

The interaction between monads and the monoidal structure of a category are of interest for probability theory: in general the probability of a product is not the product of the probabilities, and this difference encodes correlation and other statistical interaction. 
See \cite{fritzperrone2018bimonoidal} for more on this. 

So let $(\bC, \otimes, I)$ be an SMC.
A \emph{monoidal monad} is a monad $T$ together with a natural transformation of components $\nabla: TA \otimes TB\to T(A\otimes B)$ and a morphism $I\to TI$ satisfying associativity, unitality, and compatibility with the monad structure. 
The monad is \emph{symmetric monoidal} if the map $\nabla$ is compatible with permutation of the factors.
(See for example \cite[Appendix~A]{fritzperrone2018bimonoidal} for the detailed definition.)

A more general interaction between a monad and the monoidal structure of a category is the notion of \emph{strength}, widely used in theoretical computer science at least since Moggi \cite{moggi-notions-of-computation-and-monads}. 
A strength consists of a natural transformation of components $\sigma:X\otimes TY \to T(X\otimes Y)$, satisfying suitable consistency conditions (see the source above for the details). It intuitively turns a computation in $Y$ paired with a value of $X$ into a computation in $X$ and $Y$, which is in some sense ``trivial in $X$''.
A strength can always be obtained from a monoidal structure using the unit as $\nabla\circ(\eta\otimes\id):X\otimes TY \to T(X\otimes Y)$. 
Conversely, given a strength, one can obtain a monoidal structure provided that the strength satisfies a particular \emph{commutativity} condition (one speaks of a \emph{commutative monad}). 
Indeed, a symmetric monoidal structure for a monad is equivalent to a commutative strength, see for example \cite[Appendix~C]{fritz2019support} for details. 
Moreover, in that case $\Kl(T)$ is a monoidal category as well, with the tensor product induced by the one on the base category.

\section{Categories of structured objects}\label{structured}

By this we mean categories where the objects have structure and the morphisms do not necessarily preserve all of this structure.
This paper will essentially be the study of the Kleisli adjunction induced by a commutative monad on a cartesian monoidal category, but we recall here the relation to various abstractions.

\subsection{Copy-discard structure}

\begin{definition}
  Let $(\MCK,\otimes,I)$ be a symmetric monoidal category.
  A \emph{comonoid} in $\MCK$ is a triple $(X,c,d)$ where $X \in \MCK$, $c : X \to X \otimes X$, $d : X \to I$ and these satisfy the following equations.
  \begin{mathpar}
    (d \otimes 1_X) \circ c = 1_X = (1_X \otimes d) \circ c
    \and
    (c \otimes 1_X) \circ c = (1_X \otimes c) \circ c
  \end{mathpar}
  A comonoid is \emph{cocommutative} (or \emph{commutative}, for short) if in addition $c = \tau_{X,X} \circ c$.
  The category of commutative comonoids in $\MCK$, written $\CComon(\MCK,\otimes,I) = \CComon(\MCK)$, has as morphisms $(X,c,d) \to (X',c',d')$ those maps $f : X \to X'$ in $\MCK$ such that $d' \circ f = d$ and $c' \circ f = (f \otimes f) \circ c$.
\end{definition}

\begin{definition}
  A \emph{copy-delete} or \emph{copy-discard (CD) category}, also called a \emph{garbage-share (gs) monoidal category}, is a symmetric monoidal category $(\MCK,\otimes,I)$ together with a specified section of the function
  \begin{displaymath}
    \ob \CComon(\MCK,\otimes,I) \to \ob \MCK
  \end{displaymath} 
  mapping a commutative comonoid in $\MCK$ to its underlying object.
  In other words, each object $X \in \MCK$ is equipped with maps
  $\cop_X : X \to X \otimes X$ and $\del_X : X \to I$
  making it into a commutative comonoid.
\end{definition}

CD categories were first defined (in strict form, and under the name ``gs-monoidal categories'') in \cite{gadducci-thesis}, and rediscovered independently several times. See \cite[Remark~2.2]{freeGScat} (and references therein) for a more detailed history of the subject.
While every object of a CD category $\MCK$ is a comonoid, this does not make $\MCK$ a subcategory of $\CComon(\MCK)$ since the morphisms of $\MCK$ do not have to respect the comonoid structures.
This allows us to consider subclasses of morphisms in $\MCK$ which do respect them to various degrees.

\begin{example}
  Any cartesian monoidal category is a CD category in an essentially unique way. The copy map is given by the diagonal $X\to X\times X$, and the discard map is given by the unique map $X\to 1$.
\end{example}

\begin{example}
  More generally, the Kleisli category of a commutative monad $T$ on a cartesian monoidal category $\bC$ has a canonical CD structure. 
  The copy and delete structures are inherited by those of $\bC$. 
\end{example}

\begin{example}
 The Kleisli categories of $P$ and $M$, which we called $\cat{Stoch}$ and $\cat{SubStoch}$, have the following copy and discard maps.
 $$
  \begin{tikzcd}[row sep=0]
   X \ar{r} & M(X\times X) \\
   x \ar[mapsto]{r} & \delta_{(x,x)} 
  \end{tikzcd}
 \qquad 
  \begin{tikzcd}[row sep=0]
   X \ar{r} & M(1)\cong {[0,1]} \\
   x \ar[mapsto]{r} & 1_{\phantom{(x)}}
  \end{tikzcd}
 $$
 where instead $P(1)\cong 1$ for the Giry monad.
 The CD structure of $\cat{Stoch}$ has been studied in detail in \cite{chojacobs2019strings} and \cite{fritz2019synthetic}. 
\end{example}

\begin{example}
 A seemingly different example of CD category is a full subcategory of commutative comonoid objects in a symmetric monoidal category with all maps between them, not just the comonoid homomorphisms.
 Dually, we can also see this as the opposite category to a category of commutative monoid objects, for example rings or algebras.
\end{example}

\begin{definition}
 A morphism $f:X\to Y$ in a CD category is called
 \begin{itemize}
  \item \emph{copyable} if it commutes with the copy map;
  $$
  \begin{tikzcd}
   X \ar{r}{f} \ar{d}{\cop} & Y \ar{d}{\cop} \\
   X\otimes X \ar{r}{f\otimes f} & Y\otimes Y
  \end{tikzcd}
  $$
  
  \item \emph{discardable} or \emph{normalized} if it commutes with the discard map;
  $$
  \begin{tikzcd}[column sep=small]
   X \ar{rr}{f} \ar{dr}[swap]{\del} && Y \ar{dl}{\del} \\
   & I 
  \end{tikzcd}
  $$
  
  \item \emph{deterministic} if it is copyable and discardable.
 \end{itemize}
\end{definition}

A \emph{Markov category} \cite{fritz2019synthetic} is a CD category in which every morphism is normalized. 

\begin{example}
 In $\cat{SubStoch}$, a morphism $k:X\rightsquigarrow Y$ is 
 \begin{itemize}
  \item copyable if and only if for every $x\in X$, and for every measurable $B\subseteq Y$, 
 $$
 k(B|x) \in \{0, k(Y|x) \} ;
 $$
 \item normalized if and only if for every $x\in X$,
 $$
 k(Y|x) = 1 ;
 $$
 \item deterministic if and only if for every $x\in X$, and for every measurable $B\subseteq Y$, 
 $$
 k(B|x) \in \{0, 1 \} , \qquad k(Y|x) = 1 .
 $$
 \end{itemize}
 In other words, deterministic morphisms are the ones that are certain about whether any event (measurable subset $B\subseteq Y$) is going to happen (probability $1$), or not (probability $0$).
\end{example}

 Every morphism of $\cat{Stoch}$ is normalized, therefore it is a Markov category.
A canonical example of a deterministic morphism $1\rightsquigarrow Y$ of $\cat{SubStoch}$ (and $\cat{Stoch}$) is a Dirac delta $1\mapsto \delta_y$ at a point $y\in Y$. Not every deterministic morphism is in this form in general: 
\begin{example}\label{zeroone}
 Let $Y$ be the unit interval $[0,1]$, equipped with the countable-cocountable sigma-algebra (i.e.~the measurable sets are precisely the countable subsets and their complements). Then the assignment
$$
A \mapsto
\begin{cases}
 1 & \mbox{if } A \mbox{ is uncountable} \\
 0 & \mbox{if } A \mbox{ is countable}
\end{cases}
$$
is a deterministic morphism $1\rightsquigarrow Y$ of $\cat{SubStoch}$, and it cannot be written as $\delta_y$ for any $y\in Y$.
\end{example}

Note also that measures of the form $\delta_x$ do not always correspond to points bijectively, as \Cref{notinjective} shows.
One of the main purposes of this paper is, indeed, to study those deterministic kernels which are not (parametrized) Dirac deltas. More generally, to study those deterministic morphisms in a Kleisli category which do not come from morphisms of the base category (i.e.~are not pure, according to \Cref{defpure}). 

We conclude this section with a general remark.
The following conditions are equivalent for a CD category:
\begin{itemize}
 \item Every morphism is deterministic;
 \item The copy and discard maps are natural;
 \item The category is cartesian monoidal.
\end{itemize}
In some sense, one can view cartesian monoidal categories as a special case of CD categories where no randomness or nondeterminism is involved.

\subsection{Thunk-force structure}

As is a well-known fact in category theory, an adjunction between two categories gives rise to a monad on one and a comonad on the other.
Thus when $T$ is a monad on $\bC$, there is a comonad on $\Tcom$ on $\Kl(T)$.
On objects, $\Tcom X = TX$ and on morphisms $f : A \rightsquigarrow B$ we have
\begin{displaymath}
  (\Tcom f)^\sharp = TA \xrightarrow{(f^\sharp)^\dagger} TB \xrightarrow{\eta} TTB,
\end{displaymath}
where $(f^\sharp)^\dagger$ denotes the Kleisli extension of $f^\sharp$.
It is useful to describe the unit and counit in terms of some slightly richer structure.

\begin{definition}[\cite{fuhrmann-direct-models-of-the-computational-lambda-calculus}]\label{def:thunk-force-category}
  A \emph{thunk-force category} or \emph{abstract Kleisli category} is a category $\bK$ equipped with an endofunctor $L : \bK \to \bK$ and two families of maps
  \begin{mathpar}
    \thunk_A : A \to LA \and \force_A : LA \to A
  \end{mathpar}
  for $A \in \bK$ such that
  \begin{enumerate}
  \item $\force$ is a natural transformation $L \Rightarrow \id$,
  \item $\thunk_L$ is a natural transformation $L \Rightarrow LL$,
  \item $L(\thunk_A) \circ \thunk_A = \thunk_{LA} \circ \thunk_A$,
  \item $\force_A \circ \thunk_A = 1_A$,
  \item $L(\force_A) \circ \thunk_{LA} = 1_{LA}$.
  \end{enumerate}
\end{definition}
Note $\thunk$ is \emph{not} required to be natural in general.
It does follow that the endofunctor $L$ underlies a comonad with counit $\force : L \Rightarrow \id$ and comultiplication $\thunk_L : L \Rightarrow LL$.

\begin{example}\label{ex:kleisli-thunk-force}
  $\Kl(T)$ is canonically a thunk-force category. The endofunctor $L$ is given by the composite $\Tcom=F\circ G: \Kl(T)\to \bC \to \Kl(T)$, where $(F,G)$ are the functors of the Kleisli adjunction, and
  \begin{mathpar}
    (\thunk_A)^\sharp = A \xrightarrow\eta TA \xrightarrow\eta TTA
    \and
    (\force_A)^\sharp = TA \xrightarrow 1 TA.
  \end{mathpar}
  This suffices to describe the comonad structure of $\Tcom$.
\end{example}

In the context of categorical probability, i.e.~in \cite{fritz2020representable} and \cite{definetti-markov}, the morphism $\force$ is denoted by $\samp$.
This can be interpreted as a map taking a probability measure $p$ and returning a random element distributed according to $p$. 

\subsection{Thunkable morphisms}

In \cite{fuhrmann-direct-models-of-the-computational-lambda-calculus} it is shown that every thunk-force category has the form of \Cref{ex:kleisli-thunk-force}.
The crucial concept is that of \emph{thunkable} morphism, of which we recall here the basic definitions and constructions.

\begin{definition}
  A morphism $f : A \to B$ in a thunk-force category is \emph{thunkable} if the following diagram commutes.
  \begin{equation}\label{thunksquare}
  \begin{tikzcd}
   A \ar{r}{f} \ar{d}[swap]{\thunk_A} & B \ar{d}{\thunk_B} \\
   LA \ar{r}{Lf} & LB
  \end{tikzcd}
  \end{equation}
\end{definition}

The thunkable morphisms form a wide subcategory, $\bK_\thunk$.
Note that each map $\thunk_A$ is itself thunkable, by one of the axioms in \Cref{def:thunk-force-category}, as is any map of the form $Lf$, by naturality of $\thunk_L$.
The identity-on-objects inclusion $\bK_\thunk \hookrightarrow \bK$ has a right adjoint given by the factoring of $L : \bK \to \bK$ through $\bK_\thunk \hookrightarrow \bK$.
At the level of homsets the bijection is
\begin{mathpar}
  (A \xrightarrow f B) \mapsto (A \xrightarrow\thunk LA \xrightarrow{Lf} LB)
  \and
  (A \xrightarrow g LB) \mapsto (A \xrightarrow g LB \xrightarrow\force B).
\end{mathpar}
This adjunction gives rise to a monad on $\bK_\thunk$ whose Kleisli category constructed as in \Cref{ex:kleisli-thunk-force} is the original thunk-force category $\bK$.
In \cite{fuhrmann-direct-models-of-the-computational-lambda-calculus} it is further observed that in an appropriate sense this is the universal solution to the inverse problem of presenting a thunk-force category via a monad.

In \cite[Remark 3.11]{fritz2020representable} it was remarked that the unit of the monad is natural against morphisms of the base category, but not against generic Kleisli morphisms. From this perspective, we see thunkable morphisms as precisely those against which the unit of the monad is natural.

Thunkable morphisms of a Kleisli category can be characterized in terms of the base category as follows.
\begin{proposition}\label{thunkfork}
 A morphism $f:A\rightsquigarrow B$ of $\Kl(T)$ is thunkable if and only if its counterpart $f^\sharp:A\to TB$ in $\bC$ sits in the following fork.
 $$
 \begin{tikzcd}
   A \ar{r}{f^\sharp} & TB \ar[shift left]{r}{\eta} \ar[shift right]{r}[swap]{T\eta} & TTB
  \end{tikzcd}
 $$
\end{proposition}

\begin{proof}
  One way around square \eqref{thunksquare} is given by
  \begin{displaymath}
    Tf \circ \eta_{TA} \circ \eta_A = \eta_{TB} \circ f
  \end{displaymath}
  and the other one is
  \begin{displaymath}
    \mu_{TB} \circ T(\eta_{TB}) \circ T(\eta_B) \circ f = T(\eta_B) \circ f. \qedhere
  \end{displaymath}
\end{proof}

\subsection{Relationship between the different classes of maps}

Given a commutative monad on a cartesian monoidal category, we get both a copy-discard structure and a thunk-force structure canonically. 
The two structures interact in the following way. 

\begin{theorem}\label{inclusions}
 Let $T$ be a commutative monad on a cartesian monoidal category $\bC$. Consider its Kleisli category $\Kl(T)$ together with its canonical thunk-force and copy-delete structures.
  We have the following inclusions for morphisms of $\Kl(T)$,
  $$
  \mbox{pure } \subseteq \mbox{ thunkable } \subseteq \mbox{ deterministic } \subseteq \mbox{ all}.
  $$
\end{theorem}

Before proving the theorem, let's look at some examples.

\begin{example}
  Consider the `maybe' monad $TX = X+1$ on $\sets$.
  This satisfies the equalizer condition, so pure = thunkable.
  Every map is copyable, but discardable = pure.
  Therefore deterministic = thunkable.
\end{example}

\begin{example}
  The `read-only state' monad $TX = X \times X$ on $\sets$ satisfies the equalizer condition, so pure = thunkable.
  As observed in \cite{fuhrmann-direct-models-of-the-computational-lambda-calculus}, every map is copyable and discardable.
  Thus deterministic $\neq$ thunkable.
\end{example}

\begin{example}
 For the Giry monad on $\cat{Meas}$, the measure given in \Cref{zeroone} is a deterministic, but not pure morphism. However, as one can verify, applying $\eta$ and $T\eta$ one gets the same result, and so by \Cref{thunkfork}, the morphism is thunkable. So, pure $\neq$ thunkable.
 (As we will show in \Cref{Mobserv}, every deterministic morphism for the Giry monad is thunkable.)
\end{example}

Let's now prove the theorem. We make use of the following lemma, which holds in every CD category. 
\begin{lemma}[not present in the published version]\label{gfg}
  Let
  \begin{displaymath}
    \begin{tikzcd}
      A \ar{r}{f} & B \ar{r}{g} & C
    \end{tikzcd}
  \end{displaymath}
  be a composable pair of maps in any CD-category.
  \begin{itemize}
  \item [(i)] If $gf$ and $g$ are both discardable, then $f$ is discardable.
  \item [(ii)] If $gf$ and $g$ are both copyable and $g$ is split monic, then $f$ is copyable.
  \end{itemize}
\end{lemma}

\begin{proof}[Proof of \Cref{gfg}]
  For (i), we have
  \begin{align*}
    & A \xrightarrow f B \xrightarrow\del I \\
    ={} & A \xrightarrow f B \xrightarrow g C \xrightarrow\del I \\
    ={} & A \xrightarrow\del I.
  \end{align*}
  For (ii), we have
  \begin{align*}
    & A \xrightarrow f B \xrightarrow\cop B \otimes B \xrightarrow{g \otimes g} C \otimes C \\
    ={} & A \xrightarrow f B \xrightarrow g C \xrightarrow\cop C \otimes C \\
    ={} & A \xrightarrow\cop A \otimes A \xrightarrow{f \otimes f} B \otimes B \xrightarrow{g \otimes g} C \otimes C
  \end{align*}
  whence the result since $g \otimes g$ is split monic.
\end{proof}

\begin{proof}[Proof of \Cref{inclusions}]
 Let $f:X\rightsquigarrow Y$ be pure. Then there exists $g:X\to Y$ in $\bC$ such that $f^\sharp=\eta\circ g$, and so the fork of \Cref{thunkfork} can be decomposed through the unit fork \eqref{unitfork} as follows,
 $$
  \begin{tikzcd}
   X \ar{r}{g} & Y \ar{r}{\eta} & TY \ar[shift left]{r}{\eta} \ar[shift right]{r}[swap]{T\eta} & TTY
  \end{tikzcd}
 $$
 so $f$ is thunkable. 
 
 Let now $f$ be thunkable. Then $\thunk_Y\klcomp f=Tf\klcomp\thunk_X$ is equal to the composite of two pure (hence deterministic) maps.
  But as $\thunk_Y$ is deterministic and split monic, \Cref{gfg} says that $f$ is deterministic.
\end{proof}

\section{The submonad of thunkable morphisms}\label{submonad}

We now want to express the category $\Kl(T)_\thunk$ itself as the Kleisli category of a new monad, a submonad of $T$. 
Recall that every pure morphism is thunkable (\Cref{inclusions}).  
Therefore we have a functor 
$$
\begin{tikzcd}[row sep=0]
 \bC \ar{r}{i} & \Kl(T)_\thunk \\
 g \ar[mapsto]{r} & (\eta\circ g)^\flat .
\end{tikzcd}
$$
which takes $g:X\to Y$ of $\bC$ and gives the pure map $f:X\rightsquigarrow Y$ of $\Kl(T)$ such that, as above, $f^\sharp = \eta\circ g$.

Now with the help of \Cref{thunkfork}, let's construct a right-adjoint to this functor $i$.
We will suppose that the parallel pair $(\eta, T\eta)$ has an equalizer. This happens for example if $\bC$ has coreflexive equalizers (the common retraction of $\eta$ and $T\eta$ is $\mu:TTY\to TY$).
Choose an equalizer $(DY,\theta_Y)$.
$$
  \begin{tikzcd}
   DY \ar{r}{\theta_Y} & TY \ar[shift left]{r}{\eta} \ar[shift right]{r}[swap]{T\eta} & TTY
  \end{tikzcd}
  $$
  If $f:X\rightsquigarrow Y$ is thunkable then $\eta\circ f^\sharp = T\eta \circ f^\sharp$ (and conversely, this is \Cref{thunkfork}).
  Hence in this case there is a unique arrow $X\to DY$ making the triangle in the following diagram commute. 
$$
  \begin{tikzcd}
  X \ar{dr}{f^\sharp} \ar[dashed]{d} \\
   DY \ar{r}[swap]{\theta_Y} & TY \ar[shift left]{r}{\eta} \ar[shift right]{r}[swap]{T\eta} & TTY
  \end{tikzcd}
  $$
In other words, we have a bijection between thunkable morphisms $X\rightsquigarrow Y$ and arrows $X\to DY$ of $\bC$:
\begin{equation}\label{Dbijection}
\Kl(T)_\thunk(X,Y) \cong \bC(X,DY) .
\end{equation}

Since this bijection is obviously natural in $X$, we have the required adjunction.
As usual, there is a canonical extension of the object assignment $D$ to a functor $\Kl(T)_\thunk\to\bC$.
For completeness, we will describe its action of morphisms explicitly.
So now let $g:Y\rightsquigarrow Z$ be thunkable. 
Notice that, by construction, $\theta_Y:DY\to TY$ forms a fork with the pair $(\eta,T\eta)$, and so its counterpart $(\theta_Y)^\flat:DY\rightsquigarrow Y$ of $\Kl(T)$ is thunkable by \Cref{thunkfork}. 
The composition $g\circ (\theta_Y)^\flat:DY \rightsquigarrow Z$ is then thunkable too, and so its counterpart $DY\to TZ$ in $\bC$ fits into the following fork.
$$
\begin{tikzcd}
 DY \ar{r} & TZ \ar[shift left]{r}{\eta} \ar[shift right]{r}[swap]{T\eta} & TTZ
\end{tikzcd}
$$
Therefore by the universal property of the coequalizer $(DZ,\theta_Z)$, there is a unique morphism $DY\to DZ$ making the triangle in the following diagram commute. 
$$
  \begin{tikzcd}
  DY \ar{dr} \ar[dashed]{d} \\
   DZ \ar{r}[swap]{\theta_Z} & TZ \ar[shift left]{r}{\eta} \ar[shift right]{r}[swap]{T\eta} & TTZ
  \end{tikzcd}
  $$
This morphism $DY\to DZ$ is $Dg$. 
It is automatic but easy to check that $D$ preserves identities and composition, and the $\theta_Y$ assemble to form a natural transformation $D\Rightarrow T$. 

It is also plain that the bijection \eqref{Dbijection} is natural in $Y$, making $D$ a right-adjoint to $i$. 
We denote the resulting monad simply by $D$ (instead of $Di$).
Since $i$ is a left-adjoint bijective-on-objects functor, its codomain $\Kl(T)_\thunk$ is isomorphic to $\Kl(D)$.
Moreover, as $\Kl(D)\subseteq\Kl(T)$, the monad $D$ is a submonad of $T$. 

As we will see in the next section, $D$ is often idempotent.

\section{Sober objects and idempotence}\label{sober}

The original notion of \emph{sobriety} (for mathematical objects) refers to a property of topological spaces $(X,\opens(X))$.
Roughly speaking, a topological space is \emph{sober} if the existence and equality of its points is determined by the frame $\opens(X)$ of its open sets. 
On a sober topological space, a point can be uniquely identified by saying in which open sets it is contained, and conversely, any suitable consistent choice of open sets (called a \emph{completely prime filter}) corresponds to a point of the space. 
The notion of sober topological space has been linked to an equalizer of a continuation-style monad by several authors~\cite{taylor-sober-spaces-and-continuations,rosolini-equilogical-spaces-and-filter-spaces,bucalo-rosolini-sobriety-for-equilogical-spaces}. 
Our definition is less specific, and gives a notion of sobriety which in general depends on the monad.

Categorically, the points of a topological space $X$ correspond to `maps in' $1 \to X$ and the opens correspond to `maps out' $X \to \bS$ where $\bS$ is the Sierpinski space. The latter can be thought of as `observable' \cite{abramsky-domain-theory-in-logical-form} or `affirmable' \cite{vickers-topology-via-logic} properties of the points of $X$. They play a similar role to events in probability theory.
Indeed, in probability theory one faces a similar challenge: suppose we have a measure that assigns only the values $0$ or $1$ to each event (measurable set). When is this measure a Dirac delta at a unique point? Whenever this is the case for all measures, we call the space \emph{sober}, by analogy with topological spaces.
We can define the concept of sobriety in general, in terms of the monad $D$, as follows.

\begin{definition}\label{defsober}
  Let $T$ be a monad on a category $\bC$.
  An object $X \in \bC$ is \emph{sober} for the monad $T$ if its unit fork \eqref{unitfork} is an equalizer. 
\end{definition}

A similar definition was given by Taylor~\cite[Definition~4.7]{taylor-sober-spaces-and-continuations}.
See also~\cite[Section~3]{bucalo-rosolini-sobriety-for-equilogical-spaces}.

Equivalently, $X$ is sober if the unit $e:X\to DX$ of the monad $D$ is an isomorphism. That is, sober objects for $T$ are exactly the fixed points of the adjunction \eqref{Dbijection}, the one associated to $D$.
Denote by $\mathrm{Sober}(T)$ the full subcategory of $\bC$ (equivalently, of $\Kl(D)$) of sober objects. This is also known as the center of the adjunction. 

\begin{remark}
 The following conditions are equivalent for a monad $T$ on a category $\bC$:
 \begin{itemize}
  \item Every object of $\bC$ is sober.
  \item The monad $T$ satisfies Moggi's ``equalizing requirement'' \cite[Section~4]{moggi-notions-of-computation-and-monads}.
  \item Every thunkable morphism is uniquely pure. 
 \end{itemize}
\end{remark}

In particular, every object is sober whenever every deterministic morphism of $\Kl(T)$ is uniquely pure. 
For the case of Markov categories, this is in particular an instance of \emph{representability}, as defined in \cite[Section~3]{fritz2020representable}. 

\begin{example}
 For the distribution monad on $\cat{Set}$, every object is sober. 
\end{example}

\begin{example}
 For the Giry monad restricted to the category of \emph{standard Borel spaces}, every object is sober. See for example \cite[Example~10.5]{fritz2019synthetic} (where the Kleisli category is called $\cat{BorelStoch}$). 
\end{example}

Moreover, by \Cref{TXeq}, every object in the form $TX$ is sober for every monad $T$.

In general, not all objects are sober. However, in several cases, given any object we can find a universal ``sobrification''. This is the case whenever the adjunction \eqref{Dbijection} is idempotent.
By the general theory of idempotent adjunctions applied to the case of the Kleisli category of $D$, we have the following statement. 
\begin{proposition}\label{equividemp}
 The following conditions are equivalent. 
 \begin{itemize}
  \item The monad $D$ is idempotent.
  \item For every $X$, the object $DX$ is sober.
  \item The functor $D:\Kl(D)\to\bC$ is fully faithful.
  \item The inclusion of fixed points $\mathrm{Sober}(T)\hookrightarrow\Kl(D)$ is an equivalence.
  \item For every $X$, the counit $\theta^\flat:DX\rightsquigarrow X$ is an isomorphism.
 \end{itemize}
\end{proposition}

If any (hence all) of the conditions above holds, we can view $D$ as a ``sobrification'' functor, analogous to the case of topological spaces, exhibiting sober objects as a reflective subcategory of $\bC$. 

\begin{theorem}\label{thmidempotent}
 For the following categories $\bC$ and monads $T$, the associated monad $D$ is idempotent:
 \begin{itemize}
  \item The ``Giry'' monads of probability and subprobability measures $P$ and $M$ on $\cat{Meas}$;
  \item The lower Vietoris monad (a.k.a.~Hoare powerdomain) on $\cat{Top}$.
 \end{itemize}
\end{theorem}
For more details on the latter see \cite[Section~2]{fritz2019support}.

\Cref{thmidempotent} can be proven by the following helpful lemma.
\begin{lemma}\label{lem:d-idempotent-iff-t-theta-monic}
 The monad $D$ is idempotent if and only if for every $X$, the map $T\theta:TDX\to TTX$ is monic. 
 In particular, as the map $\theta:DX\to TX$ is an equalizer, it suffices to show that $T$ maps regular monomorphisms to monomorphisms.
\end{lemma}

\begin{proof} 
 First of all, by \Cref{equividemp}, $D$ is idempotent if and only if for every $X$, the counit $\theta^\flat:DX\rightsquigarrow X$ is an isomorphism.
 
 Denote now by $e:X\to DX$ the unit of the monad $D$, which can be obtain from the universal property of $D$ as an equalizer for the unit fork, as in the following diagram.
 $$
 \begin{tikzcd}
 X \ar[dashed]{d}[swap]{e} \ar{dr}{\eta} \\
 DX \ar{r}[swap]{\theta} & TX \ar[shift left]{r}{\eta} \ar[shift right]{r}[swap]{T\eta} & TTX 
 \end{tikzcd}
 $$
 By the triangle identities of the adjunction of $D$, $\theta^\flat$ is split epi with section given by the map $E\coloneqq (\eta\circ e)^\flat:X\rightsquigarrow DX$ induced by unit $e:X\to DX$.
 Therefore $\theta^\flat$ is an isomorphism if and only if $E$ is its actual inverse, i.e.~if $E\klcomp \theta^\flat=\id_{DX}$, which in terms of the category $\bC$ reads
 \begin{equation}\label{Tthetamonic}
 DX \xrightarrow \theta TX \xrightarrow{Te} TDX  = DX \xrightarrow \eta TDX .
 \end{equation}
 So suppose that \eqref{Tthetamonic} holds. Then 
 \begin{align*}
 & TDX \xrightarrow{T\theta} TTX  \xrightarrow{TTe} TTDX \xrightarrow{\mu} TDX \\
 & = TDX \xrightarrow{T\eta} TTDX \xrightarrow{\mu} TDX \\
 & = TDX \xrightarrow{\id} TDX ,
 \end{align*}
 so $T\theta$ is split monic.
 
 Conversely, suppose that $T\theta$ is monic. Then 
 \begin{align*}
    DX \xrightarrow \theta TX \xrightarrow{Te} TDX \xrightarrow{T\theta} TTX
    & = DX \xrightarrow \theta TX \xrightarrow{T\eta} TTX \\
    & = DX \xrightarrow \theta TX \xrightarrow{\eta} TTX \\
    & = DX \xrightarrow \eta TDX \xrightarrow{T\theta} TTX ,
  \end{align*}
 which implies \eqref{Tthetamonic}. 
\end{proof}

\begin{proof}[Proof of \Cref{thmidempotent}]
 By \Cref{lem:d-idempotent-iff-t-theta-monic}, it suffices to show that the Giry monad maps regular monomorphisms of $\cat{Meas}$ (i.e.~embeddings of measurable spaces) to monomorphisms of $\cat{Meas}$ (i.e.~injective measurable functions).
 So let $i:X\to Y$ be an embedding of measurable spaces, that is, an injective function such that the every measurable subset $A$ of $X$ is in the form $f^{-1}(B)$ for some measurable subset $B$ of $Y$. Let $p$ and $q$ be measures on $X$, and suppose that $f_*p=f_*q$. Then for every measurable subset $A\subseteq X$ we can find a measurable $B\subseteq Y$ such that
 $$
 p(A) = p(f^{-1}(B)) = f_*p(B) = f_*q(B) = q(f^{-1}(B)) = q(A) .
 $$
 Therefore, $p=q$ already on $X$, and hence $f_*:MX\to MY$ and its restriction $PX\to PY$ are injective.
 
 The lower Vietoris monad case is analogous, once one sees closed sets as dual to open sets (as in \Cref{exampleH}).
\end{proof}

The idempotent monad associated to the lower Vietoris monad is not only idempotent, but it is also the sobrification monad of topology, hence the name ``sober''. 
\begin{theorem}\label{indeedsober}
 The monad $D$ associated to the lower Vietoris monad $H$ on $\cat{Top}$ is the functor assigning to a topological space $X$ the subset of $HX$ given by the irreducible closed sets. 
\end{theorem}

The proof of this theorem is given at the end of \Cref{exampleH}.
The result resembles the known characterizations of sobriety in terms of equalizers~\cite{taylor-sober-spaces-and-continuations,rosolini-equilogical-spaces-and-filter-spaces,bucalo-rosolini-sobriety-for-equilogical-spaces}, but note that $HX$ is not quite a continuation (it behaves more like a subspace of the continuation --- but of course $\cat{Top}$ is not cartesian closed, see also \cite[Appendix~B]{fritz2019support}).

\section{Observational monads}\label{observational}

In higher-order programming languages one typically has a `ground type' such as `nat' or `bool', representing actual data that we can handle as input or output, as well as higher-order types (i.e.\ function types), or open terms, which are never directly accessible as inputs or outputs and only appear at intermediate stages of computation.
Thus it is of interest to consider when two values of a higher-order type can be interchanged in the middle of programs without altering the \emph{observable} behaviour of the computer.
From the point of view of probability theory one has the same intuition for \emph{real numbers}, i.e.~proving equality of probability measures involves, in the end, proving that certain integrals give the same number.
An \emph{observational monad} makes this intuition precise: in the sense proposed here it corresponds to a semantics of open terms which is abstract for a certain kind of observable equivalence. 

Let $\bC$ be a cartesian monoidal category and $T$ a commutative monad.
Since each $X \in \bC$ is a commutative comonoid in $\Kl(T)$, for each $n \in \bN$ there is a canonical map
\begin{displaymath}
  \cop_n : TX \rightsquigarrow (TX)^{\otimes n}
\end{displaymath}
obtained by iterating the copy map (all possibilities are equal by coassociativity).

By post-composing with $\force \otimes \ldots \otimes \force : (TX)^{\otimes n} \rightsquigarrow X^{\otimes n}$, we define a map
\begin{equation}\label{eq:samp-n}
  \samp_n : TX \rightsquigarrow X^{\otimes n},
\end{equation}
the \emph{$n$'th sampling map}.
As special cases, $\samp_0 = \del$ and $\samp_1 = \force$. (Note that the map $\samp$ appearing in \cite{fritz2020representable} and \cite{definetti-markov} corresponds to our $\samp_1 = \force$.)

We now want to make the intuition precise that \emph{probability measures can be tested for equality by taking repeated independent samples}.

\begin{definition}\label{def:observational-monad}
  Let $T$ be a commutative monad on a cartesian monoidal category $\bC$.
  Then $T$ is an \emph{observational monad} if for every object $X$ the family of maps $(\samp_n : TX \rightsquigarrow X^{\otimes n})_{n \in \bN}$ is jointly monic in $\Kl(T)$.
\end{definition}

\begin{remark}
  The notion of observational monad also makes sense for strong monads that are not necessarily commutative, but we will not pursue that line here.
\end{remark}

\begin{proposition}[Remark~6.3 in the published version]\label{monicinC}
The maps \eqref{eq:samp-n} are jointly monic in $\Kl(T)$ iff the maps
\begin{equation}\label{eq:samp-n-in-C}
  TTX \xrightarrow{T(\Delta_n)} T((TX)^n) \xrightarrow{T(\nabla_n)} TT(X^n) \xrightarrow \mu T(X^n)
\end{equation}
are jointly monic in $\bC$.
\end{proposition}

Intuitively, the map in \eqref{eq:samp-n-in-C} takes a distribution on distributions on $X$, samples to get a distribution on $X$, and then returns the result of $n$ independent samples from that distribution.

\begin{proof}
 Consider the following diagram in $\Kl(T)$. 
 $$
 \begin{tikzcd}[column sep=small]
  A \ar[shift left, kl]{r}{f} \ar[shift right, kl]{r}[swap]{g} & TX \ar[kl]{rr}{\cop_n} && TX^{\otimes n} \ar[kl]{rrr}{\thunk^{\otimes n}} &&& X^{\otimes n}
 \end{tikzcd}
 $$
 The condition that the maps $\samp_n$ form a monic family in $\Kl(T)$ means that whenever the composites $\thunk^{\otimes n}\klcomp\cop_n\klcomp f$ and $\thunk^{\otimes n}\klcomp\cop_n\klcomp g$ in the diagram above are equal for all $n$, then $f=g$.
 
 Now in terms of the category $\bC$ we can rewrite the (Kleisli) composites above as follows.
 $$
 \begin{tikzcd}[column sep=small]
  A \ar[shift left]{r}{f^\sharp} \ar[shift right]{r}[swap]{g^\sharp} & TTX \ar{rr}{T(\Delta_n)} && T((TX)^n) \ar{rr}{T(\nabla_n)} && TT(X^n) \ar{r}{\mu} & T(X^n) 
 \end{tikzcd}
 $$
 The condition that the maps $\samp_n$ form a monic family in $\Kl(T)$ now reads: whenever the composites in this new diagram are equal, then $f=g$, or equivalently $f^\sharp=g^\sharp$. But this means precisely that the maps in \eqref{eq:samp-n-in-C} are jointly monic in $\bC$.
\end{proof}

From \eqref{eq:samp-n-in-C} it is clear that taking one sample $n=1$ is in general insufficient: that would correspond to hoping that the map $\mu:TTX\to TX$ alone is monic. This is in general not the case. For the Giry monad, for example, this is far from injective: a probability measure can be in general obtained as a mixture of other measures in several different ways. 

The two major consequences of observationality will be given in \Cref{definetti} and \Cref{detmaps}. 
In the rest of this section, we give some technical sufficient conditions that one can use in order to prove that a monad is observational, and we introduce the idea of ``objects which are directly observable'' (such as real numbers for probability).

\subsection{Observations via result objects}\label{resultobjects}

An observational monad is one for which an `observation procedure' for the powers of $X$ can be transferred to one for $TX$.
Ordinarily, direct observations are only made for \emph{ground types}, e.g.\ Boolean values or real numbers, and all observations at higher types ultimately implemented in terms of direct ground observations.
We call these special types the \emph{result objects}.

\begin{definition}\label{def:s-observational-monad}
  Let $R$ be an object of $\bC$.
  We say that the monad $T$ is \emph{$R$-observational} iff for each object $X$, the family of morphisms
  \begin{equation}\label{eq:s-observational}
        TX \xrsquigarrow{\samp_n} X^{\otimes n} \xrsquigarrow{h_1^\flat \otimes \ldots \otimes h_n^\flat} R^{\otimes n}
  \end{equation}
  where $n \in \bN$ and $h_1,\ldots,h_n : X \to TR$, is jointly monic.
  We call $R$ the \emph{result object}.
\end{definition}

It is easy to see that if $T$ is $R$-observational then $T$ is observational.
We can also write the $R$-observationality condition in terms of the category $\bC$, analogously to \Cref{monicinC}. It reads that the following maps need to be jointly monic for all $n$.
 \begin{equation}\label{RobjC}
  TTX \xrightarrow{T(Th_1,\dots,Th_n)} T((TTR)^n) \xrightarrow{T(\nabla_n\nabla_n)} TTT(R^n) \xrightarrow{\mu\mu} T(R^n)
 \end{equation}

As well as being useful for demonstrating that a given monad is observational, the property of being $R$-observational allows us to exploit \Cref{lem:d-idempotent-iff-t-theta-monic}.
Recall that an object $P$ is \emph{$\MCM$-injective} \cite{maclane-cwm1998} with respect to a class $\MCM$ of morphisms iff whenever $(i : A \to B) \in \MCM$ and $h : A \to P$, there exists a (not necessarily unique) map $k : B \to P$ with $k \circ i = h$.

\begin{proposition}
  Suppose that $T$ is $R$-observational and that $TR$ is injective with respect to the class of regular monomorphisms.
  Then the submonad $D$ of thunkable morphisms is idempotent.
\end{proposition}
\begin{proof}
  By \Cref{lem:d-idempotent-iff-t-theta-monic} it suffices to show that $T\theta_X$ is monic.
  Since $\eta_{TTX} \circ T\theta_X = TT\theta_X \circ \eta_{TDX}$ it suffices for this to show that $TT\theta_X$ is monic.
  (Indeed, this is also necessary, since we saw above that if $T\theta_X$ is monic then it is also split monic).
  The $R$-observationality condition for $DX$ says that the family of maps \eqref{RobjC} is monic, so it suffices to show that
  \begin{align*}
    &\mu\mu\circ T(\nabla_n\nabla_n) \circ T(Th_1,\dots,Th_n) \circ TT\theta_X \\
    &= \mu\mu\circ T(\nabla_n\nabla_n) \circ T(T(h_1\circ\theta_X),\dots,T(h_n\circ\theta_X))
  \end{align*}
  is monic.
  By $R$-observationality, it is sufficient that every map $s : DX \to TR$ factorize as $h \circ \theta_X$ for some $h : TX \to TR$, since then this family is equivalent to the one in the $R$-observationality condition for $TX$.
  But this just says that $TR$ is a $\{\theta_X\}$-injective object.
\end{proof}

\subsection{Result objects with a monoid structure}

A particularly useful case of result object is when $R$ is a monoid in $\bC$. Then $TR$ is a monoid in the category of $T$-algebras, with unit $e^{TR}$ and multiplication $m^{TR}$ given as follows, 
\begin{equation}\label{TRmonoid}
1 \xrightarrow{\eta} T1 \xrightarrow{Te^R} TR \qquad
TR \times TR \xrightarrow{\nabla} T(R\times R) \xrightarrow{Tm^R} TR
\end{equation}
where $e^R$ and $m^R$ are the unit and multiplication of $R$.
For probability monads, we generally take $TR$ to be the unit interval $[0,1]$, which is a monoid under multiplication (as well as an algebra under integration). In particular,
\begin{itemize}
 \item For the distribution and Giry monad $P$, the object $[0,1]$ can be written as $P\{0,1\}$ (where the product in $R=\{0,1\}$ is multiplication);
 \item For the monad of subprobability measures $M$, the object $[0,1]$ can be written as $M1$ (with the trivial monoid structure on $R=1$).
\end{itemize}

Let now $(S,s:TS\to S)$ be a $T$-algebra with monoid structure (for example, in the form $S=TR$ as above, with $R$ a monoid in $\bC$.) The Eilenberg-Moore adjunction gives a bijection for each object $X$,
$$
\begin{tikzcd}[row sep=0]
 \bC(X,S) \ar{r}{\cong} & \cat{Alg}(T)(TX,S) \\
 f \ar[mapsto]{r} & \varepsilon_f = s\circ Tf .
\end{tikzcd}
$$
For the Giry monad, for example, given the function $f:X\to [0,1]$, the corresponding morphism $\varepsilon_f:PX\to [0,1]$ is the \emph{integral} of $f$:
$$
p \mapsto \int f\, dp .
$$

Because of this correspondence, we can test observationality of $T$ in terms of the maps $\varepsilon_f$. Now, these maps alone are in general not enough to test observationality, but their \emph{products} are. 
Let's define what we mean by ``product''. Let's write the unit and multiplication of the monoid $S$ by $e^S:1\to S$ and $m^S:S\times S\to S$. 
Denote now by $m^S_n:S^n\to S$ the maps given by 
\begin{itemize}
 \item $m_0=e$;
 \item $m_1=\id$;
 \item $m_2=m$;
 \item For $m>2$, $m_n$ is the unique (by associativity) way of multiplying $n$ objects, $S\times\dots\times S\to S$. 
\end{itemize}

Given now $h_1,\ldots,h_n : X \to S$, define the \emph{pointwise product} $h_1\cdots h_n$ as the map 
\begin{equation}\label{ptwiseprod}
X \xrightarrow{\Delta_n} X^n \xrightarrow{h_1\times\ldots\times h_n} S^n \xrightarrow{m_n} S .
\end{equation}
In $\cat{Meas}$, if the $h_i$ are functions into $[0,1]$, this gives the ordinary product of functions
$x\mapsto h_1(x)\cdots h_n(x)$.

We can now test observationality by means of the maps $\varepsilon_f$ and their products.
\begin{lemma}\label{usingmonoid}
 Let $R$ be a monoid in $\bC$, and consider the free algebra $(S,s)=(TR,\mu)$ with its induced monoid structure.
 The monad $T$ is $R$-observational (hence observational) if the following maps are jointly monic,
 \begin{equation}\label{dualmaps}
   TTX \xrightarrow{T(\varepsilon_{h_1}\cdots \varepsilon_{h_n})}  TS \xrightarrow{s} S
 \end{equation}
 for all $n \in \bN$ and $h_1,\ldots,h_n : X \to S$.
\end{lemma}
Note that we can write the maps above even more concisely as $\varepsilon_{\varepsilon_{h_1}\cdots \varepsilon_{h_n}}$. 

\begin{proof}
 It suffices to show that the maps of
 the family \eqref{dualmaps} can be obtained from the maps of family \eqref{RobjC} by postcomposition (because if the former are jointly monic, then surely the latter have to be as well). 

 Now since $S=TR$ and $s=\mu_R$, we have that $\varepsilon_{h_i}=\mu_R\circ Th_i$, and so we can rewrite \eqref{dualmaps} as follows,
 \begin{align*}
 &\mu \circ T(\mu\circ Th_1\cdots \mu\circ Th_n) \\
 &= \mu \circ T(m^S_n) \circ T(\mu\circ Th_1,\dots, \mu\circ Th_n) \\
 &= \mu \circ TT(m_n^R) \circ T\nabla_n \circ T(\mu^n) \circ T(Th_1,\dots, Th_n) \\
 &= Tm_n^R\circ\mu \circ T\nabla_n \circ T(\mu^n) \circ T(Th_1,\dots, Th_n) \\
 &= Tm_n^R\circ\mu \circ T\mu \circ TT\nabla_n \circ T\nabla_n \circ T(Th_1,\dots, Th_n) \\
 &= Tm_n^R\circ\mu \circ \mu \circ TT\nabla_n \circ T\nabla_n \circ T(Th_1,\dots, Th_n)
 \end{align*}
 where we used, in turn, 
 \begin{itemize}
  \item the pointwise product formula \eqref{ptwiseprod};
  \item the fact that $m_n^S=Tm_n^R \circ\nabla$, from \eqref{TRmonoid};
  \item naturality of $\mu$;
  \item compatibility of the monoidal structure $\nabla$ of $T$ with $\mu$;
  \item the associativity square for $\mu$.
 \end{itemize}
 The last line is exactly the composition of $Tm_n^R$ with \eqref{RobjC}.
\end{proof}

\Cref{usingmonoid} allows us to prove observationality of several monads, including the Giry monad (see \Cref{examples}).

\section{Determinism in the observational case}\label{detmaps}

The key consequence of a monad's being observational is the following.

\begin{theorem}\label{main}
  Let $T$ be an observational commutative monad on a cartesian monoidal category $\bC$.
  Then every deterministic morphism in $\Kl(T)$ is thunkable.
\end{theorem}
\begin{proof}
  Let $f : A \rightsquigarrow B$ be deterministic.
  It suffices to check that 
  \begin{mathpar}
    A \xrsquigarrow f B \xrsquigarrow\thunk TB
    \and
    A \xrsquigarrow\thunk TA \xrsquigarrow{T^\flat f} TB
  \end{mathpar}
  become equal after postcomposition with each of the $\samp_n$ maps from \eqref{eq:samp-n}.
  For $n = 1$,
  \begin{align*}
    & \samp_1 \klcomp \thunk \klcomp f \\
    ={} & \force \klcomp \thunk \klcomp f\\
    ={} & f \\
    ={} & f \klcomp \force \klcomp \thunk  \\
    ={} & \force \klcomp T^\flat f \klcomp \thunk \\
    ={} & \samp_1 \klcomp T^\flat f \klcomp \thunk .
  \end{align*}
  For $n = 0$, we use the facts that $f$ is discardable and $\thunk$ and $T^\flat(f)$ are both pure (and so discardable):
  \begin{align*}
    & \samp_0 \klcomp \thunk \klcomp f \\
    ={} & \del \klcomp \thunk \klcomp f \\
    ={} & \del \\
    ={} & \del \klcomp T^\flat f \klcomp \thunk \\
    ={} & \samp_0 \klcomp T^\flat f \klcomp \thunk .
  \end{align*}
  For $n \geq 2$, we use the facts that $f$ is copyable and that $\thunk$ and $T^\flat(f)$ are both pure (and so copyable) and the already proved $n = 1$ case.
  \begin{align*}
    & \samp_n \klcomp \thunk \klcomp f \\
    ={} & \force^{\otimes n} \klcomp \cop_n \klcomp \thunk \klcomp f \\
    ={} & \force^{\otimes n} \klcomp \thunk^{\otimes n} \klcomp \cop_n \klcomp f \\
    ={} & \force^{\otimes n} \klcomp \thunk^{\otimes n} \klcomp f^{\otimes n} \klcomp \cop_n \\
    ={} & \force^{\otimes n} \klcomp (T^\flat f)^{\otimes n} \klcomp \thunk^{\otimes n} \klcomp \cop_n \\
    ={} & \force^{\otimes n} \klcomp \cop_n \klcomp T^\flat f \klcomp \thunk \\
    ={} & \samp_n \klcomp T^\flat f \klcomp \thunk. \qedhere
  \end{align*}
\end{proof}

\begin{corollary}\label{thunkdet}
 Let $T$ be an observational commutative monad on a cartesian category. The associated submonad $D$ is equivalently characterizing the \emph{deterministic} morphisms of $\Kl(T)$. 
\end{corollary}

This will be the case for example for the Giry monad, as we show in in \Cref{exampleM}.

\section{De Finetti's theorem}\label{definetti}

De Finetti's theorem \cite{de-finetti-foresight-its-logical-laws-its-subjective-sources} gives a connection between \emph{random distributions} and \emph{exchangeable sequences}.
As we show in this section, the notion of observationality can say something about de Finetti's theorem, as infinite sequences can be seen as limits of a finite, arbitrarily large amount of observations.

In general, in a Kleisli category we have monoidal products rather than cartesian ones. 
Because of that, in order to talk about infinite sequences in $X$, one cannot take a countable cartesian product of copies of $X$. One, rather, has to extend monoidal products to the infinite case.
This was accomplished in \cite{fritzrischel2019zeroone} for the case of Markov categories. Here we give the analogous construction for CD categories (which is almost the same). For further context, motivation, and applications we refer to the aforementioned source, as well as to the later \cite{definetti-markov}.

\subsection{Kolmogorov products}

In order to form Kolmogorov products, let's take a look at the so-called \emph{marginalization maps}. Given objects $X$ and $Y$ in a CD category, their tensor product
$X\otimes Y$ can be interpreted as the object of \emph{joint states}. For example, if we are in a Kleisli category, a morphism $I\rightsquigarrow X\otimes Y$ corresponds to an arrow $1\to T(X\times Y)$ of the base category, which is a joint probability measure if $T$ is the Giry monad. 
We can then apply the map $ X\otimes Y \xrsquigarrow{\id\otimes\del} X\otimes I \cong X $, which intuitively ``discards'' $Y$. This maps a joint state into the \emph{marginal} state on $X$. For the case of the Giry monad, this corresponds to taking the marginal probability.

Since $\otimes$ is a bifunctor, marginalizations are deterministic, and also natural in the sense that the following diagram commutes.
$$
\begin{tikzcd}
 X\otimes Y \ar[kl]{r}{\id\otimes\del} \ar[kl,swap]{d}{\del\otimes\id} & X\otimes I \ar[kl]{d}{\del\otimes\id} \\
 I\otimes Y \ar[kl,swap]{r}{\id\otimes\del} & I
\end{tikzcd}
$$
The same is true for marginalizations of finite sequences $X_1\otimes\dots\otimes X_n$. 

\begin{definition}
 Let $I$ be an infinite set, and let $\{X_i\}$ be an $I$-indexed collection of objects of a CD category. 
 The \emph{Kolmogorov} product of the family $\{X_i\}$, denoted by $X_I$, is the (cofiltered) limit of the diagram whose objects are the finite tensor products 
 $ \bigotimes_{i\in F} X_i $
 over all the finite subsets $F$ of $I$, and whose morphisms are the marginalizations, if moreover the following two further conditions are satisfied:
 \begin{itemize}
  \item The arrows of the limit cone are deterministic;
  \item The limit is preserved by the tensor product $X\otimes -$ for each object $X$. 
 \end{itemize}
\end{definition}

We call a \emph{Kolmogorov power} a Kolmogorov product where the objects $X_i$ are all isomorphic to a same object $X$. In that case we denote the product by $X^I$.
In the context of de Finetti, we are interested in \emph{countable Kolmogorov powers}, i.e.~where $I$ is, equivalently, the set $\bN$ of natural numbers.

Let's now look at the interaction between observational monads and Kolmogorov powers. 
If the Kolmogorov product $A^\bN$ exists, the maps $\samp_n$ make the following diagram commute for each $n$ and for each marginal projection $\pi:A^{\otimes n}\rightsquigarrow A^{\otimes n-1}$, i.e.~$\pi\klcomp\samp_n=\samp_{n-1}$. 
Therefore, there exist a unique map $\samp_{\bN}$ making the following diagram commute for each $n$.
$$
\begin{tikzcd}[row sep=tiny]
    &  A^{\bN} \ar[kl]{dd}{\pi} \\ 
    TA \ar[kl]{ur}{\samp_{\bN}} \ar[kl]{dr}[swap]{\samp_{n}} \\
    &  A^{\otimes n}
  \end{tikzcd}
$$
This map is monic if and only if the family $\{\samp_n\}$ is jointly monic. 
Therefore, if $\Kl(T)$ has countable Kolmogorov products, $T$ is observable if and only if $\samp_{\bN}$ is monic for every object $A$.

\subsection{Application to the synthetic de Finetti theorem}

In \cite{definetti-markov} a de Finetti theorem was proved for Markov categories satisfying particular conditions. It reads as follows.
\begin{theorem}[Theorem 4.4~in\cite{definetti-markov}]
 Consider a Markov category with countable Kolmogorov powers, conditionals, and almost-surely-compatibly representable, with probability monad $P$.\footnote{For the precise definition of these properties, we refer to the original source.}
 Then a morphism $p: A\rightsquigarrow X^\bN$ is exchangeable if and only if there is a morphism $\rho:A\rightsquigarrow PX$ such that 
 $$
 p = \samp_\bN \klcomp \rho .
 $$
\end{theorem}

We now know that if (and only if) the probability monad $P$ is observational, then $\samp_\bN$ is monic. Therefore, in the theorem above, we can even conclude that the morphism $\rho$ is unique. 
This is in particular the case for the Giry monad on standard Borel spaces, and so for the Markov category $\cat{BorelStoch}$.
In traditional probability theory, the uniqueness of $\rho$ (at least for states, $A=I$) is already known. Our formalism, however, incorporates this statement of uniqueness into the categorical formalism. 
This opens the road to study de Finetti's theorem and similar statements as categorical universal properties, a path undertaken already (using a different formalism) for example by \cite{definetti-jacobs-staton}.

\section{Examples}\label{examples}

Here we give some important examples of observational monads: the Giry and sub-Giry monads and the lower Vietoris monad of nondeterminism.
We also give an interesting but unusual example with the name generation monad.

\subsection{The Giry monads on measurable spaces}\label{exampleM}

\begin{theorem}\label{Mobserv}
  The Giry monad and the monad $M$ of subprobability measures are observational. 
\end{theorem}

The proof uses the following version of the celebrated functional monotone class theorem, (a.k.a.~the ``$\pi$-$\lambda$ theorem for functions''). See \cite[Theorem~2.12.9]{bogachev} for a reference.

\begin{theorem}[Functional monotone class theorem]\label{monclass}
 Let $X$ be a measurable space, and denote by $F(X)$ the space of bounded measurable functions $X\to\bR$. 
 Consider a vector subspace $H\subseteq F(X)$ containing the function $1$, and such that its positive cone is closed under sequential increasing limits. 
 Consider a subset $K\subseteq H$ closed under pointwise products, and denote by $\sigma(K)$ the $\sigma$-algebra generated by the functions in $K$. 
 Then $H$ contains all $\sigma(K)$-measurable functions.
\end{theorem}

In order to prove \Cref{Mobserv}, let's also recall that given a measurable set $X$, the (``Giry'') $\sigma$-algebra of $MX$ is equivalently generated by the functions 
$$ 
\begin{tikzcd}[row sep=0]
 MX \ar{r}{\varepsilon_f} & {[0,1]} \\
 m \ar[mapsto]{r} & \int f \, dm
\end{tikzcd}
$$
for each measurable function $f:X\to[0,1]$ (see~\cite{giry}). 

Let's now prove our theorem.

\begin{proof}[Proof of \Cref{Mobserv}]
Let's prove the assert for $M$, the case of $P$ is analogous. We will use \Cref{usingmonoid} with $S=[0,1]$ (i.e.~result object $R=1$). So it suffices to show that the maps \eqref{dualmaps} are jointly monic. 
Let's unpack the expression \eqref{dualmaps} for our case. We need to show that on a measurable space $X$, if $p$ and $q$ are measures \emph{on} $MX$, i.e.~\emph{in} $MMX$, then $p=q$ if and only if $p$ and $q$ agree on the following products,
 \begin{equation}\label{agreeonproducts}
  \int_{MX} \varepsilon_{h_1} \cdots \varepsilon_{h_n} \, dp = \int_{MX} \varepsilon_{h_1} \cdots \varepsilon_{h_n} \, dq 
 \end{equation}
 for each finite collection $\{h_1,\dots,h_n\}$ of measurable functions $X\to[0,1]$, including the 0-ary product (the function $1$, meaning that $p$ and $q$ have the same normalization).
 
 Now consider the collection $K$ of all finite pointwise products of the functions $\varepsilon_f$,
$$
K \coloneqq \left\{ \varepsilon_{f_1} \cdots \varepsilon_{f_n} : n\in\bN, f_i\in F(X), \lambda\in\bR \right\} .
$$
This collection generates the sigma algebra of $MX$, and it is closed under products. 
Take now the measures $p,q$ \emph{on} $MX$, i.e.~\emph{in} $MMX$. Denote by $H$ the subset of bounded measurable functions $g:MX\to\bR$ such that $\int g \, dp = \int g\, dq$. 
By linearity and monotone continuity of integration, $H$ is a vector subspace of $F(MX)$ and its positive cone is closed under sequential increasing limits. Moreover, it contains $1$.
Since $p$ and $q$ agree on $K$ (by \eqref{agreeonproducts}), $K\subseteq H$, and so we are in the hypothesis of \Cref{monclass}. The theorem tells us that every function that is measurable for the $\sigma$-algebra generated by $K$ lies in $H$, i.e.~cannot tell $p$ and $q$ apart. But since $K$ generates the whole $\sigma$-algebra of $MX$, this means that $p=q$.  
\end{proof}

\begin{corollary}
  The (sub-)Giry monad admits an idempotent submonad of deterministic states.
\end{corollary}

\begin{proof}
 We know from \Cref{thmidempotent} that the monad $D$ of thunkable morphisms is an idempotent submonad of $M$ (resp.~$P$).
 Since $M$ (resp.~$P$) is observational, we know by \Cref{main} that thunkable and deterministic morphisms coincide. 
 Therefore we can equivalently view $D$ as the monad whose Kleisli morphisms are deterministic (i.e.~zero-one) Markov kernels. In particular, the elements of $DX$ are the zero-one measures on $X$.
\end{proof}

\subsection{The lower Vietoris monad on topological spaces}\label{exampleH}

\begin{theorem}\label{Hobserv}
 The lower Vietoris monad on $\cat{Top}$ (a.k.a.~the Hoare powerdomain) is observational. 
\end{theorem}

Again as result object we take the terminal object $1$, so that $S=H1$ is the Sierpinski space $\{0,1\}$, equipped with its usual topology (generated by $\{1\}$).
This way, a continuous function $f:X\to\{0,1\}$ is equivalently and open subset of $X$ (by taking $f^{-1}(1)$). 

We use the following statement, \cite[Lemma~2.3]{fritz2019support}.

\begin{proposition}\label{pilambdaH}
 Let $X$ be a topological space. Let $\bB$ be a basis of the topology of $X$. Let $C$ and $D$ be closed subsets of $X$, i.e.~elements of $HX$. Then $C=D$ if and only if for every open $U$ in $\bB$, the set $C$ intersects $U$ if and only if $D$ does.
\end{proposition}

Moreover, the topology of $HX$ is the weakest topology making the following maps continuous for all open sets $U$ of $X$, 
\begin{equation}\label{basisHX}
\begin{tikzcd}[row sep=0]
 HX \ar{r}{\varepsilon_U} & \{0,1\} \\
 C \ar[mapsto]{r} & \begin{cases}
                     1 \quad C\cap U \ne \varnothing ;\\
                     0 \quad C\cap U = \varnothing .
                    \end{cases}
\end{tikzcd}
\end{equation}
As maps into $\{0,1\}$ corresponds to open sets, we can view the $\varepsilon_U$ as open sets of $HX$ (generating the topology). 
For more details, see again \cite[Section~2]{fritz2019support}.

Let's now prove the theorem. In some sense, the role of the $\pi$-$\lambda$ theorem this time is played by \Cref{pilambdaH}.

\begin{proof}[Proof of \Cref{Hobserv}]
 We use \Cref{usingmonoid} where $S$ is the Sierpinski space, with monoid structure given by `meet', and $H$-algebra structure, i.e.~(topological) sup-semilattice structure, given by `join'.
 It suffices to show that the maps \eqref{dualmaps} are jointly monic.
 Unpacking the expression of \Cref{usingmonoid}, and using \eqref{basisHX}, given closed subsets $C,D\subseteq HX$ we have to show that they are equal if and only if for all $n$ and for all open sets $U_1,\subseteq U_n\subseteq X$, the set $C$ intersects the intersection
 \begin{equation}\label{intersections}
 \varepsilon_{U_1} \cap \cdots \cap \varepsilon_{U_n}
 \end{equation}
 if and only $D$ does. Now, as the topology of $HX$ is generated by the $\varepsilon_U$, the $\varepsilon_U$ for a subbasis, and hence their intersection form a basis. By \Cref{pilambdaH}, then, the sets \eqref{intersections} are indeed enough to test that $C=D$. 
\end{proof}

In particular, Kleisli morphisms for the monad $H$ are deterministic if and only if they are thunkable. 

We are now ready to prove \Cref{indeedsober}, i.e.~that for the monad $H$, sober objects are exactly sober topological spaces.

\begin{proof}[Proof of \Cref{indeedsober}]
 Let $X$ be a topological space, and let $C\in HX$ be a closed subset of $X$. We have to prove that $C$ is irreducible if and only if it is thunkable as a morphism $1\to HX$ of $\cat{Top}$, or equivalently, deterministic.
 Now first of all, $C$ is discardable as a Kleisli morphism of $H$ if and only if $C$, as a set, is nonempty. 
 Moreover, $C$ as a morphism is copyable if and only if the closed subsets $C\times C$ and $\cop(C)=\{(c,c): c\in C\}$ of $X\times X$ are equal. 
 By \Cref{pilambdaH} we can test equalities of closed subsets of $X\times X$ by looking at a basis, and we pick the basis given by the products $U\times V$ of open subsets $U,V\subseteq X$.
 We have that $C\times C$ intersects $U\times V$ if and only if $C$ intersects $U$ and $V$ separately, and that $\cop(C)$ intersects $U\times V$ if and only if $C$ intersects $U$ and $V$ at the same point, i.e.~if $C\cap U\cap V$ is nonempty. In other words,
 $$
 (C\times C) \cap (U\times V) = \cop(C) \cap (U\times V) ,
 $$
 for all $U$ and $V$, and so $C\times C$ and $\cop(C)$ are equal, if and only if the mapping on open sets
 $$
 \begin{tikzcd}[row sep=0]
 O(X) \ar{r} & \{0,1\} \\
 U \ar[mapsto]{r} & \begin{cases}
                     1 \quad C\cap U \ne \varnothing ;\\
                     0 \quad C\cap U = \varnothing .
                    \end{cases}
\end{tikzcd}
$$
preserves binary intersections. That is, $C$ induces a completely prime filter on the frame $O(X)$, which means precisely that $C$ is irreducible.
\end{proof}

\subsection{Local names}

In \cite{sabok-staton-stein-wolman-probabilistic-programming-semantics-for-name-generation} it was shown that randomness can be used to model fresh name generation, exploiting the fact that repeated uniform sampling from $[0,1]$ returns distinct values with probability 1.
Conversely, the traditional model of local names from \cite{stark-categorical-models-for-local-names} is an interesting model of categorical probability.
We take $\cat D = [\inj,\sets]$ and the \emph{name-generation monad}
\begin{displaymath}
  (TX)a = \colim_{b \in \inj} X(a + b).
\end{displaymath}
If a `stage of computation' means the (finite) set of names which are in use so far, an object $X \in [\inj,\sets]$ consists of a set of values at each stage.
A value of $TX$ at stage $a$ consists of a set $b$ of local names that have been generated together with a value of $X$ at the resulting stage $a+b$.
The object $TX$ is quotiented so that: 1) the names in $b$ are `bound' or $\alpha$-convertible; 2) names that are not referenced in the value of $X$ are out of scope and so discarded.

The monad $T$ is commutative and affine (satisfies $T1 \cong 1$).
It does not satisfy the equalizing requirement, but it is observational.

\begin{theorem}
  The monad $T$ on $[\inj,\sets]$ is observational.
\end{theorem}

\begin{proof}
  Let $[b,[c,x]]$ and $[b',[c',x']]$ be representatives of elements of $(TTX)a$ which are equal under all of the $\samp_n$ maps.
  Then, for every $n \in \bN$, writing $c_1,\ldots,c_n$ and $c'_1,\ldots,c'_n$ for tuples of copies of $c$ and $c'$ using fresh names and $x_i = x[c_i/c]$, $x'_i = x'[c'_i/c']$, we have that
  \begin{displaymath}
    (b + c_1 + \ldots + c_n , (x_1,\ldots,x_n))
  \end{displaymath}
  and
  \begin{displaymath}
    (b' + c'_1 + \ldots + c'_n , (x'_1,\ldots,x'_n))
  \end{displaymath}
  are equivalent elements of $(TX)(a)$, meaning that there are injections $f : b + c_1 + \ldots + c_n \hookrightarrow d$ and $f' : b' + c'_1 + \ldots + c'_n \hookrightarrow d$ such that $X(1_a + f)(x_i) = X(1_a + f')(x'_i) \in X(a + d)$ for each each $1 \leq i \leq n$.
  By choosing $n$ sufficiently large, i.e.\ bigger than $|b| + |b'|$, we can ensure that for some $i$ the image of $c_i$ under $f$ is disjoint from the image of $b'$ under $f'$ and also that the image of $c'_i$ under $f'$ is disjoint from the image of $b$ under $f$.
  Let $d_1 = \im(f|_b) \cup \im(f'|_{b'})$ and $d_2 = d \setminus d_1$.
  Then $[b,[c,x]]$ is also represented by $[d_1,[d_2 ,X(1_a + f|_b + f|_{c_1 + \ldots + c_n})(x_i)]]$ and similarly $[b',[c',x']]$ is also represented by $[d_1,[d_2,X(1_a + f'|_b + f'|_{c'_1 + \ldots c'_n})(x'_i)]]$, but by construction these are actually equal.
\end{proof}

It is well known that the full subcategory of $[\inj,\sets]$ whose objects are the pullback-preserving functors is equivalent to the \emph{Schanuel topos}, or topos of \emph{nominal sets}~\cite{murdoch-pitts-a-new-approach-to-abstract-syntax-with-variable-binding}.
The following facts are straightforward to establish.
\begin{enumerate}
\item Every object $TX$ is a nominal set.
\item $T$ preserves monos whose codomain is a nominal set.
\item Every nominal set is sober with respect to $T$.
\end{enumerate}

\begin{corollary}
  The sober objects for the name-generation monad are precisely the nominal sets, and the sobrification monad is idempotent.
\end{corollary}

\section{Conclusion}\label{conclusion}

\subsection{Related work}

The notion of \emph{sober space} comes from topology, meaning a topological space where the set of points (set of values) is determined by the set of open subsets (observable predicates).
The idea that abstractly a `sober' object is one for which the fork \eqref{unitfork} is an equalizer has appeared before, e.g.~\cite{taylor-sober-spaces-and-continuations,rosolini-equilogical-spaces-and-filter-spaces,bucalo-rosolini-sobriety-for-equilogical-spaces}.
These works were not focused on probability, but there is some overlap in the examples of interest.
In \cite{taylor-sober-spaces-and-continuations}, it is shown that the passage $\bC \mapsto \Kl(T)_\thunk$ `freely adjoins sobriety' to $\bC$.
Our concern is a little different: we use thunkable morphisms to transform the \emph{objects} of $\bC$ into sober objects in the same category.

The concepts of discardability, copyability, and thunkability have been developed, for example, in \cite{thielecke-continuation-semantics-and-self-adjointness,fuhrmann-direct-models-of-the-computational-lambda-calculus,fuhrmann-varieties-of-effects,kammar-plotkin-algebraic-foundations-for-effect-dependent-optimisations}.
Since most work was on general computational effects, there is another fundamental class of morphisms of interest, the \emph{central morphisms}.
For commutative effects as studied in this paper, every morphism is central.
Thus in general one needs to consider \emph{symmetric premonoidal categories} \cite{levy-power-thielecke-modelling-environments-in-cbv-programming-languages}, rather than symmetric monoidal categories.
We note that much of our framework does not rely on monoidality rather than premonoidality, but we leave to future work the investigation of noncommutative examples.

\subsection{Summary}

We have given natural conditions on a commutative monad, \emph{observationality} (Def.~\ref{observational}) and \emph{$S$-observational} (Def.~\ref{def:s-observational-monad}), for which the deterministic computations are precisely the thunkable ones.
Under mild conditions we showed that these imply that the monad has an associated, idempotent \emph{sobrification} submonad (\S\ref{submonad}, \S\ref{sober}).
We showed that these conditions apply to several examples of interest, including the Giry monad on measurable spaces, and the lower Vietoris monad on topological spaces (\Cref{thmidempotent}).

\paragraph{Acknowledgements} 
We would like to thank Tobias Fritz, Tomáš Gonda, and Dario Stein, as well as Sam Staton and all his research group, for the interesting and fruitful conversations and feedback. 
We also would like to thank the anonymous reviewers for their helpful remarks.

\vskip 10cm

\bibliographystyle{ACM-Reference-Format}
\bibliography{./markov}

\end{document}